\documentclass[12pt]{article}

\usepackage{fullpage}
\usepackage[T1]{fontenc}
\usepackage{amssymb} 
\usepackage{amsmath}  
\usepackage{amsthm}

\usepackage{xcolor}
\usepackage[normalem]{ulem} 

\usepackage{hyperref}

\title{\bf Gamma Acyclicity, Annotated Relations, and Consistency Witness Functions}

\author{Albert Atserias  \\
Universitat Polit\`ecnica de Catalunya \\
Centre de Recerca Matem\`atica \\
Barcelona, Catalonia, Spain 
\and
Phokion G.\ Kolaitis \\
UC Santa Cruz \& IBM Research \\
Santa Cruz, California, USA
}

\begin{document}

\maketitle

\newcommand\phk[1]{\textcolor{blue}{(Phokion) #1}}
\newcommand\albert[1]{\textcolor{green}{(Albert) #1}}


\newcommand{\phrcomment}[1]{\marginpar{\tiny \textbf{M:} #1}}
\newcommand{\aacomment}[1]{\marginpar{\tiny \textbf{A:} #1}}

\newtheorem{lemma}{Lemma}
\newtheorem{claim}{Claim}
\newtheorem{corollary}{Corollary}
\newtheorem{theorem}{Theorem}
\newtheorem{proposition}{Proposition}
\newtheorem{definition}{Definition}
\newtheorem{fact}{Fact}



\newenvironment{examplebf}{\refstepcounter{example}\par\bigskip
\noindent\textbf{Example~\theexample.} \rmfamily}{\hfill$\dashv$\medskip}

\newcommand{\supp}{{\mathrm{Supp}}}
\newcommand{\tuples}{{\mathrm{Tup}}}
\newcommand{\domain}{{\mathrm{Dom}}}
\newcommand{\Entropy}{{\mathrm{H}}}
\newcommand{\KL}{{\mathrm{D}}}
\newcommand{\ftp}{transportation property}
\newcommand{\icp}{inner consistency property}
\newcommand{\ltgc}{local-to-global consistency property}

\newcommand{\gcpr}{global consistency problem for relations}
\newcommand{\gcpb}{global consistency problem for bags}
\newcommand{\glcpr}{{\sc GCPR}}
\newcommand{\glcpb}{{\sc GCPB}}

\newcommand{\intcone}{\mathrm{intcone}}
\newcommand{\norm}[1]{\Vert #1 \Vert}
\newcommand{\bnorm}[1]{\norm{#1}_{\mathrm{b}}}
\newcommand{\unorm}[1]{\norm{#1}_{\mathrm{u}}}
\newcommand{\suppnorm}[1]{\norm{#1}_{\mathrm{supp}}}
\newcommand{\munorm}[1]{\norm{#1}_{\mathrm{mu}}}
\newcommand{\mbnorm}[1]{\norm{#1}_{\mathrm{mb}}}

\newcommand{\free}[1]{\mathbb{F}(#1)}

\newcommand{\transpose}{\mathrm{T}}
\newcommand{\Id}{\mathrm{I}}
\newcommand{\trace}{\mathrm{tr}}

\newcommand{\onto}{\stackrel{\scriptscriptstyle{s}}\rightarrow}

\newcommand{\standardjoin}[1]{\Join_{#1,\mathrm{S}}}
\newcommand{\vorobyevjoin}[1]{\Join_{#1,\mathrm{V}}}
\newcommand{\componentwisejoin}[2]{\Join^{#1}_{#2}}

\newcommand{\commentout}[1]{}

\newcommand{\icpadj}{innerly consistent}

\newcommand{\newatop}[2]{\genfrac{}{}{0pt}{2}{#1}{#2}}

\newcommand{\wit}{{\mathrm Wit}}

\newcommand{\cj}{\Join_{\mathrm c}}

\newcommand{\cjoin}{${\mathrm c}$-join}

\newcommand{\umon}{uniformly monotone}








\thanks{}

\begin{abstract}

During the early days of relational database theory it was realized that ``acyclic'' database schemas possess a number of desirable  semantic  properties. In fact, three different notions of ``acyclicity'' were identified and extensively investigated during the 1980s, namely,  $\alpha$-acyclicity, $\beta$-acyclicity,  and $\gamma$-acyclicity. Much more recently, the study of $\alpha$-acyclicity was extended to annotated relations, where the annotations are values from some positive commutative monoid. 
The recent results about $\alpha$-acyclic schemas and annotated relations give rise  to
results about $\beta$-acyclic schemas and annotated relations, since a schema is $\beta$-acyclic if and only if every sub-schema of it is $\alpha$-acyclic.
Here, we study $\gamma$-acyclic schemas and annotated relations. Our main finding is that the desirable semantic properties of $\gamma$-acyclic schemas extend to annotated relations, provided the annotations come from a positive commutative monoid that has the transportation property. Furthermore, the results reported here shed light on the role of the join of two standard relations, Specifically, our results reveal that
the only relevant property of the join of two standard relations is that it is a witness to the consistency of the two relations, provided that these two relations are consistent. For the more abstract setting of annotated relations, this property of the standard join is captured by the notion of a consistency witness function, a notion which we systematically investigate in this work.

\end{abstract}

\maketitle

\newpage

\section{Introduction} \label{sec:intro}

Annotated databases are databases in which each fact in a relation is annotated with a value from some algebraic structure. Starting with the influential work on database provenance \cite{DBLP:conf/pods/GreenKT07,DBLP:journals/sigmod/KarvounarakisG12}, there has been an extensive investigation of several different aspects of annotated databases, including 
the study of conjunctive query containment for annotated databases \cite{DBLP:journals/mst/Green11,DBLP:journals/tods/KostylevRS14} and  the evaluation of Datalog programs on annotated databases \cite{DBLP:journals/jacm/KhamisNPSW24}. In these investigations, the annotations are values  from  some fixed semiring ${\mathbb K}=(K,+, \times, 0,1)$. Thus, standard relational databases are annotated databases  in which the annotations are $1$ (true) and $0$ (false), while bag databases are annotated databases in which the  annotations are non-negative integers denoting the multiplicities. This framework, which is often referred to as \emph{semiring semantics}, has spanned  first-order logic \cite{DBLP:journals/corr/abs-1712-01980} and least fixed-point logic
\cite{DBLP:conf/csl/DannertGNT21}.

During the early days of relational database theory it was realized that ``acyclic'' database schemas possess a number of desirable  semantic  properties. In fact, three different notions of ``acyclicity'' were identified and extensively investigated during the 1980s, namely,  acyclicity (also known as $\alpha$-acyclicity), $\beta$-acyclicity,  and $\gamma$-acyclicity. On undirected graphs (equivalently, on database schemas consisting of 
binary relation symbols only) these notions coincide with the notion of an acyclic graph, but they form a strict hierarchy on hypergraphs (equivalently, on arbitrary database schemas) with $\beta$-acyclicity  being a stricter notion than acyclicity, and $\gamma$-acyclicity  being a stricter notion than  $\beta$-acyclicity.

The study of acyclic schemas was initiated by Yannakakis, who focused on the evaluation of acyclic joins \cite{DBLP:conf/vldb/Yannakakis81}.  After this, Fagin, Beeri, Maier Yannakakis \cite{BeeriFaginMaierYannakakis1983}
showed that acyclic schemas are precisely the ones possessing the \emph{\ltgc}, that is, every collection of pairwise consistent relations $R_1,\ldots, R_m$ over such schemas is globally consistent (i.e.,  there a relation $T$ whose projection on the attributes of $R_i$ is equal to $R_i$, for $1 \leq i \leq m$). Fagin et al.\ \cite{BeeriFaginMaierYannakakis1983} also characterized acyclicity in terms of the existence of \emph{monotone sequential join expressions}, i.e., expressions of the form 
$(((\cdots (R_1\cj R_2) \Join \cdots )\Join R_{m-1} )\Join R_m)$ with the property that if the relations $R_1,\ldots,R_m$ are pairwise consistent, then every intermediate sequential join expression $((\cdots (R_1\cj R_2) \Join \cdots )\Join R_{i-1})$ produces a relation that is consistent with the relation $R_i$.
Results about acyclicity yield results about $\beta$-acyclicity, since a schema is $\beta$-acyclic if and only if every sub-schema of it is acyclic. Fagin \cite{DBLP:journals/jacm/Fagin83} then
studied  $\gamma$-acyclicity  and showed that a schema is $\gamma$-acyclic if and only every \emph{connected} sequential join expression is monotone. Intuitively,  this means  that every sequential join expression is monotone, provided no join between relations with disjoint sets of attributes is allowed.

Atserias and Kolaitis \cite{AK25} studied the interplay between local consistency and global consistency for  annotated relations. Since the definition of consistency of  annotated relations uses only the projection operation on relations and since projection is defined using only addition $+$, they focused on $\mathbb K$-relations, i.e., annotated relations in which the annotations come from a monoid ${\mathbb K=(K, +, 0)}$.
 They identified a condition on monoids,  called the \emph{transportation property}, and showed that a positive monoid ${\mathbb K}=(K,+,0)$ has the transportation property if and only if every acyclic schema $H$ has the \ltgc~for $\mathbb K$-relations (i.e., every pairwise consistent collection of $\mathbb K$-relations over $H$ is globally consistent). It was not clear, however, whether  the results about acyclic schemas and sequential join expressions in \cite{BeeriFaginMaierYannakakis1983} can be extended to annotated relations, since, as shown in \cite{DBLP:conf/pods/AtseriasK21}, the analog of the standard join for bags need not  be a witness to the consistency of two consistent bags.
 In a subsequent paper, Atserias and Kolaitis \cite{DBLP:journals/sigmod/AtseriasK25} 
introduced the notion of a \emph{consistency witness function} on a positive monoid ${\mathbb K}$, which is a function $W$ that, given two $\mathbb K$-relations $R$ and $S$, returns a $\mathbb K$-relation $W(R,S)$ that is a consistency witness for $R$ and $S$, provided that $R$ and $S$ are consistent $\mathbb K$-relations; they also introduced the notion of a \emph{monotone sequential \cjoin~expression}, which is analogous to that of a monotone sequential join expression with some arbitrary consistency witness function in place of the standard join. Using these notions, it was shown in \cite{DBLP:journals/sigmod/AtseriasK25} that the characterization of acyclicity in terms of monotone sequential join expressions in \cite{BeeriFaginMaierYannakakis1983} extends  to  characterizations of acyclicity in terms of monotone sequential \cjoin~expressions on monoids  possessing the transportation property; furthermore, the transportation property itself can be characterized in such terms.

Here, we investigate $\gamma$-acyclic schemas and establish that
 the desirable semantic properties of $\gamma$-acyclic schemas extend to annotated relations.
 The two main results  are as follows:
    \begin{enumerate}
\item If $\mathbb K$ is a positive commutative monoid and $H$ is a schema 
such that every connected sequential \cjoin-expression over $H$ is monotone on $\mathbb K$ w.r.t.\ \emph{some} consistency witness function on $\mathbb K$, then $H$ is $\gamma$-acyclic.
\item If $\mathbb K$ is a positive commutative monoid that has the transportation property and $H$ is a $\gamma$-acyclic schema, then 
every connected sequential \cjoin-expression over $H$ is monotone on $\mathbb K$
w.r.t.\ \emph{every} consistency witness function on $\mathbb K$.
    \end{enumerate}
As a byproduct of these two main results, we obtain a characterization of the transportation property in terms of $\gamma$-acyclicity and connected sequential \cjoin~expressions. 
 Furthermore, our work sheds light on the role of the join of two standard relations, Specifically, our results reveal that, in the study of the various notions of acyclicity in \cite{DBLP:journals/jacm/Fagin83},
the only relevant property of the join of two standard relations is that it is a witness to the consistency of the two relations, provided that these two relations are consistent. In the setting of  annotated relations, this property of the standard join is captured by the notion of a consistency witness function.

 The rest of the paper is organized as follows. Section 2 contains the definitions of the basic notions, while Section 3 contains the definition of a consistency witness function and related notions. To make the paper as self-contained as possible,  the earlier results about acyclic schemas are summarized in Section 4. Section 5 discusses $\beta$-acyclic schemas. Section 6 contains the main results about $\gamma$-acyclic schemas and annotated relations.

\section{Basic Notions} \label{sec:prelims}

\paragraph{Monoids} 
A \emph{commutative monoid} is a structure $\mathbb{K}=(K,+,0)$, where $+$ is a binary operation on the universe $K$ of $\mathbb K$  that is
 associative, commutative, and has $0$ as its neutral  element, i.e., $p+ 0 = p = 0 + p$ holds for all $p\in K$.  
 A commutative monoid $\mathbb K = (K,+,0)$ is \emph{positive} 
 if for all elements $p,q\in K$ with
$p+q=0$, we have that  $p=0$ and $q=0$.  From now on, we  assume that all commutative monoids considered have at least two elements in their universe.

The following are examples of positive commutative monoids.
\begin{itemize}
\item The \emph{Boolean monoid}
 $\mathbb{B} = (\{0,1\},\vee,0)$ with disjunction $\vee$ as its operation and~$0$~(false) as its neutral element.
 \item The \emph{bag  monoid}  $\mathbb{N}=(Z^{\geq 0}, +, 0)$,  where $Z^{\geq 0}$ is the set of non-negative integers   and $+$ is the standard addition operation. 
 Note that the structure ${\mathbb Z}=(Z,+,0)$, where $Z$ is the set of integers, is a commutative monoid, but not a positive one. 
 \item A \emph{numerical semigroup} is a submonoid ${\mathbb K}=(K,+,0)$ of  the bag monoid $\mathbb{N}=(Z^{\geq 0}, +, 0)$, such that $K$ is a cofinite set, i.e., the complement $Z^{\geq 0}\setminus K$ is  finite. A concrete example of a numerical semigroup is ${\mathbb K}=(\langle 3,5\rangle, +,0)$, where $\langle 3,5\rangle$ is  the set of all non-negative integers of the form $3m+5n$ with  $m\geq 0$ and $n\geq 0$,  i.e., $\langle 3,5\rangle=\{0,3,5,6,8, 9, 10, \ldots \}$.
 \item 
  The structures ${\mathbb T}= (R\cup\{\infty\}, \min, \infty)$ and  ${\mathbb V}=([0,1], \max, 0)$, where $R$ is the set of all real numbers, $[0,1]$ is the interval of all real numbers between $0$ and $1$, and $\min$ and $\max$ are the standard minimum and maximum operations.
  \item The \emph{power set monoid} ${\mathbb P}(A)= (\mathcal{P}(A), \cup,\emptyset)$, where 
  if $A$ is a set, then $\mathcal{P}(A)$ is its powerset, and $\cup$ is the union operation on sets.
\end{itemize}

\paragraph{$\mathbb K$-relations and marginals of $\mathbb K$-relations}
An \emph{attribute}~$A$ is a symbol with an associated
set~$\domain(A)$ as its \emph{domain}. If~$X$ is a finite set of
attributes, then~$\tuples(X)$ is the set
of~\emph{$X$-tuples}, i.e., the set of
functions that take each attribute~$A \in X$ to an element of its
domain~$\domain(A)$. $\tuples(\emptyset)$ is non-empty as it
contains the \emph{empty tuple}, i.e., the  function with empty
domain. If~$Y \subseteq X$  and~$t$ is
an~$X$-tuple, then the \emph{projection of~$t$ on~$Y$}, denoted
by~$t[Y]$, is the unique~$Y$-tuple that agrees with~$t$ on~$Y$. In
particular,~$t[\emptyset]$ is the empty tuple.

Let~${\mathbb K} = (K,+,0)$ be a positive commutative monoid and let~$X$ be a finite set
of attributes.
\begin{itemize}
    \item 
A~\emph{$\mathbb{K}$-relation over~$X$} is a
function~$R : \tuples(X) \rightarrow K$ that assigns a value~$R(t)$ in~$K$
to every~$X$-tuple~$t$ in~$\tuples(X)$. 
We will often write $R(X)$ to indicate that $R$ is a $\mathbb K$-relation over $X$, and we will refer to $X$ as the set of attributes of $R$.

If~$X$ is the empty set of attributes, then a~$\mathbb{K}$-relation 
over~$X$ is simply a function that assigns a  single value from~$K$  to the empty tuple. 
Note that the $\mathbb B$-relations are  the standard relations, while the $\mathbb N$-relations are the \emph{bags} or 
\emph{multisets}, i.e., each tuple has a non-negative integer associated with it that denotes the  \emph{multiplicity} of the tuple.
\item 
The \emph{support} $\supp(R)$ of
a~$\mathbb{K}$-relation~$R(X)$ is the set
of~$X$-tuples~$t$ that are assigned non-zero value, i.e.,
$\supp(R) := \{ t \in \tuples(X) : R(t) \not= 0 \}$.
We will often write~$R'$ to
denote~$\supp(R)$. Note that~$R'$ is a standard relation
over~$X$. A~$\mathbb{K}$-relation is \emph{finitely supported} if its support is a
finite set. In this paper, all~$\mathbb{K}$-relations considered will be  finitely supported, 
and we omit the term; thus, from now on, a $\mathbb{K}$-relation is a finitely supported $\mathbb{K}$-relation. 
When~$R'$ is empty, we say that~$R$ is the empty~$\mathbb{K}$-relation over~$X$. 
\item 
If~$Y\subseteq X$, then the \emph{marginal $R[Y]$ of $R$ on $Y$}  is the~$\mathbb{K}$-relation 
over~$Y$ such that for every~$Y$-tuple~$t$, we have that
   \begin{equation}
R[Y](t) := \sum_{\newatop{r \in R':}{r[Y] = t}} R(r).
\label{eqn:marginal}
\end{equation}
The value $R[Y](t)$ is  the \emph{marginal of $R$ over $t$}. For notational simplicity, we will often write $R(t)$ for the marginal of $R$ over $t$, instead of $R[Y](t)$. It will be clear from the context (e.g., from the arity of the tuple $t$) if $R(t)$ is indeed the marginal of $R$ over $t$ (in which case $t$ must be a $Y$-tuple) or $R(t)$ is the actual value of $R$ on $t$ as a mapping from $\tuples(X)$ to $K$ (in which case $t$ must be an $X$-tuple).
   Note that if $R$ is a standard  relation (i.e., $R$ is a $\mathbb B$-relation), then the marginal $R[Y]$ is the projection of $R$ on $Y$.
\end{itemize}
   The proof of the next useful proposition follows easily from the definitions.

\begin{proposition} \label{lem:easyfacts1} 
Let $\mathbb K$ be a 
 positive commutative monoid and let $R(X)$ be a  $\mathbb K$-relation. Then the
  following  hold:
  \begin{enumerate} \itemsep=0pt
  \item For all~$Y \subseteq X$, we have~$R'[Y] = R[Y]'$.
  \item For all $Z \subseteq Y \subseteq X$, we have $R[Y][Z] = R[Z]$.
\end{enumerate}
\end{proposition}

\commentout{
\begin{proof}
  For the first part, the inclusion~$R[Y]' \subseteq R'[Y]$ is obvious. For the converse, assume that~$t \in R'[Y]$, so
  there exists~$r$ such that~$R(r) \not= 0$ and~$r[Y] =
  t$. By~\eqref{eqn:marginal} and the positivity of~$\mathbb K$, we have
  that~$R(t) \not= 0$. Hence~$t \in R[Y]'$.  
  
  For the second part, we have
  \begin{equation}
  R[Y][Z](u) = \sum_{\newatop{v \in R[Y]':}{v[Z]=u}} R[Y](v) =
  \sum_{\newatop{v \in R'[Y]:}{v[Z]=u}} \sum_{\newatop{w \in R':}{w[Y]=v}} R(w) =
  \sum_{\newatop{w \in R':}{w[Z]=u}} R(w) = R[Z](u)
\end{equation}
where the first equality follows from~\eqref{eqn:marginal}, the second
follows from the first part of this lemma to replace~$R[Y]'$ by~$R'[Y]$, and
again~\eqref{eqn:marginal}, the third follows from partitioning the
tuples in~$R'$ by their projection on~$Y$, together
with~$Z \subseteq Y$, and the fourth follows from~\eqref{eqn:marginal}
again.
\end{proof}
}

\noindent If~$X$ and~$Y$ are sets of attributes, then we write~$XY$ as
shorthand for the union~$X \cup Y$. 
\commentout{
Accordingly, if~$x$ is
an~$X$-tuple and~$y$ is a~$Y$-tuple such
that~$x[X \cap Y] = y[X \cap Y]$, then we write~$xy$ to denote
the~$XY$-tuple that agrees with~$x$ on~$X$ and on~$y$ on~$Y$.  We
say that~\emph{$x$ joins with~$y$}, and that~\emph{$y$ joins
  with~$x$}, to \emph{produce} the tuple~$xy$.}


\paragraph{Schemas and hypergraphs}
\begin{itemize}
\item A \emph{schema} is a sequence~$X_1,\ldots,X_m$ of non-empty sets of attributes.  
\item A \emph{hypergraph} is a pair $H=(V,F)$, where $V$ is a finite non-empty set and $F$ is a set of non-empty subsets of $V$. We call $V$ the set of the \emph{nodes} of $H$ and we call $F$ the set of the \emph{hyperedges} of $H$.
\end{itemize}
A schema~$X_1,\ldots,X_m$ can be identified with  the hypergraph $H=(\bigcup_{i=1}^m X_i,\{X_1,\ldots,X_m\})$, i.e.,  the nodes of $H$  are the attributes  and  the hyperedges of $H$ are the members~$X_1,\ldots,X_m$ of the schema.
We will use the terms \emph{schema} and
\emph{hypergraph} interchangeably.
\begin{itemize}
\item A  \emph{collection of
$\mathbb{K}$-relations} over  a  schema~$X_1,\ldots,X_m$
is a sequence $R_1(X_1),\ldots,R_m(X_m)$  such that each~$R_i(X_i)$ is a~$\mathbb{K}$-relation over~$X_i$.
\end{itemize}

\section{Consistent Relations and Consistency Witness Functions}  \label{sec:cons}

Let $\mathbb K=(K,+,0)$ be a positive commutative monoid.
\begin{itemize}
\item Two $\mathbb K$-relations $R(X)$ and $S(Y)$ are \emph{consistent} if there is a $\mathbb K$-relation $T(XY)$ such that $T[X]=R$ and $T[Y]=S$.
Such a $\mathbb K$-relation $T$ is  a \emph{consistency witness} for $R$ and $S$.
\item A collection $R_1(X_1),\ldots,R_m(X_m)$ of
$\mathbb{K}$-relations over a  schema
$X_1,\ldots,X_m$ is \emph{globally consistent} if there is  a $\mathbb K$-relation $T(X_1\ldots X_m)$ such that
$T[X_i]=R_i$, for $i$ with $1\leq i\ \leq m$. Such a $\mathbb K$-relation $T$ is  a \emph{consistency witness} for $R_1,\ldots,R_m$.
\end{itemize}
Note that if $R_1(X_1),\ldots,R_m(X_m)$
is a globally consistent collection of $\mathbb K$-relations, then the relations~$R_1(X_1),\ldots,R_m(X_m)$ are pairwise consistent. Indeed, if $T$ is a consistency witness for
$R_1(X_1),\ldots,R_m(X_m)$, then for all $i$ and $j$ with $1\leq i,j \leq m$, we have that the $\mathbb K$-relation $T[X_iX_j]$ is a consistency witness for $R_i$
and $R_j$, because 
$$
R_i =T[X_i]=T[X_iX_j][X_i] \quad \mbox{and} \quad  
R_j  =T[X_j]=T[X_iX_j][X_j],$$
where, in each case, the first equality follows from the definition of global consistency and the second equality follows from Proposition \ref{lem:easyfacts1}. 

It is well known that the converse  fails even for standard relations, i.e., there are standard relations that are pairwise consistent but not globally consistent. The main result by Beeri et al.~\cite{BeeriFaginMaierYannakakis1983} characterizes the schemas for which the pairwise consistency of a collection of standard relations implies that they are globally consistent. 
Quite recently, this result was extended 
 to $\mathbb K$-relations over positive monoids that satisfy a condition called the \emph{transportation property} \cite{AK25}.
We will discuss this extension in the next section. For now, we consider 
the following notion, which was introduced and studied in
\cite{DBLP:journals/sigmod/AtseriasK25}.
\begin{itemize}
    \item A \emph{consistency witness function on $\mathbb K$} is a function $W$ that takes as arguments two  $\mathbb K$-relation $R(X)$  and  $S(Y)$, and returns as value a $\mathbb K$-relation $W(R,S)$ over  $XY$ such that if $R$ and $S$ are consistent $\mathbb K$-relations, then $W(R,S)$ is a consistency witness for $R$ and $S$. 
    \end{itemize}
Clearly, the standard join operation $\Join$ of standard relations is a consistency witness function on the Boolean monoid $\mathbb B$. As pointed out in \cite{DBLP:conf/pods/AtseriasK21}, however, the bag-join operation of bags (the  analog of the standard join for bags) is not a consistency witness function on the bag monoid~$\mathbb N$.

A \emph{join expression} is an expression involving standard relations and applications of the join operation $\Join$ on standard relations \cite{BeeriFaginMaierYannakakis1983, DBLP:journals/jacm/Fagin83}. In \cite{DBLP:journals/sigmod/AtseriasK25}, the notion of a \emph{\cjoin~expression} and its variants were introduced as a generalization of the notion of join expression to arbitrary consistency witness functions and $\mathbb K$-relations. The precise definitions 
are as follows.

Let $X_1,\ldots,X_m$ be a schema
and let $\cj$ be a binary function symbol, which will be interpreted by some consistency witness function.
\begin{itemize}
\item 
The collection of \emph{\cjoin~expressions over $X_1,\ldots,X_m$} is the smallest collection of strings that contains each $X_i$ and has the property that if $E_1$ and $E_2$ are in the collection, then also the string $(E_1\cj E_2)$ is  in the collection.

\item The collection of \emph{sequential \cjoin~expressions over $X_1,\ldots,X_m$} is the smallest collection of strings that contains each $X_i$ and has the property that if $E$ is  in the collection and $X$ is one of the $X_i$'s, then also the string $(E \cj X)$ is in the collection.

\end{itemize}
Note that a (sequential) \cjoin~expression over $X_1,\ldots,X_m$ need not contain every set $X_i$.

Clearly, the string 
 $((X_1\cj X_2) \cj X_3)$  is a sequential \cjoin-expression, while the string
 $((X_1\cj X_2)\cj (X_3\cj X_4))$ is a \cjoin~expression, but not a sequential one.

 Semantics to \cjoin~expressions are assigned in a straightforward way as follows

Let $X_1,\ldots,X_m$ be a schema and let $E$ be a \cjoin-expression over $X_1,\ldots,X_m$. 
If    $W$ is a consistency witness function on $\mathbb K$
and  $R_1(X_1),\ldots,R_m(X_m)$ is a collection of $\mathbb K$-relations, we write    $E(W,R_1,\ldots,R_m)$ to denote the $\mathbb K$-relation over
$X_1\cdots X_m$ obtained  by evaluating $E$ when $\cj$ is interpreted by $W$ and each $X_i$ is interpreted by $R_i$ for $i = 1,\ldots,m$. 

The next notion yields a sufficient condition for a \cjoin~expression to give rise to global consistency witnesses.
\begin{itemize}
\item Let $E$ be a \cjoin~expression over a schema $X_1,\ldots,X_m$, let $W$ be a consistency witness function on $\mathbb K$, and let~$R_1(X_1),\ldots,R_m(X_m)$ be a collection of $\mathbb K$-relations.
We say that $E$ is \emph{monotone with respect to  $W$ and  $R_1(X_1),\ldots,R_m(X_m)$} if for every sub-expression $E_1 \cj E_2$ of $E$, we have that the $\mathbb K$-relations $E_1(W,R_1,\ldots,R_m)$
and $E_2(W,R_1,\ldots,R_m)$ are consistent. 
\end{itemize}
The next proposition from \cite{DBLP:journals/sigmod/AtseriasK25} is proved in a straightforward way   by induction on the construction of \cjoin~expressions and by using Proposition \ref{lem:easyfacts1}.
\begin{proposition} \label{prop:monotone-general}
Let  $E$ be a \cjoin~expression over $X_1,\ldots,X_m$, let $W$ be a
consistency witness function on $\mathbb K$, and let $R_1(X_1),\ldots,R_m(X_m)$ be $\mathbb K$-relations. If $E$ is monotone with respect to $W$ and $R_1,\ldots,R_m$, and  every $X_i$ occurs in $E$, then 
$E(W,R_1,\ldots,R_m)$ is a global consistency witness for the $\mathbb K$-relations $R_1,\ldots,R_m$. 
\end{proposition}
\commentout{
This proposition is proved by induction on the construction of \cjoin~expressions.

The base case is trivial, since in this case $E$ is $X_i$ for some $i$ with $1\leq i\leq n$, hence $E(W,R_i)=R_i$, which is a consistency witness for $R_i$.

For the inductive step, assume that $E$ is $E_1 \cj E_2$, where $E_1$ and $E_2$ are \cjoin~expressions for which the inductive hypothesis holds.  To simplify the notation, let us put ${\bf R}=(R_1,\ldots,R_m)$; furthermore, we  put
${\bf R}_1=(R_i : i \in I_1)$ and ${\bf R}_2=(R_i : i \in I_2)$, where $I_1$ and $I_2$ are
the sets of indices $i$ such that  $X_i$  occurs in $E_1$ and in $E_2$, respectively. 
In this case, we have that  $E(W,{\bf R})=W(E_1(W,{\bf R}_1), E_2(W,{\bf R}_2))$.  

Since $E$ is monotone with respect to $W$ and ${\bf R}$, we have that the $\mathbb K$-relations $E_1(W,{\bf R}_1)$ and $E_2(W,{\bf R}_2)$ are consistent, hence $W(E_1(W,{\bf R}_1), E_2(W,{\bf R}_2))$ is a consistency witness for $E_1(W,{\bf R}_1)$ and $E_2(W,{\bf R}_2)$.
We must  show that 
$W(E_1(W,{\bf R}_1), E_2(W,{\bf R}_2))[X_i]=R_i$ holds, for every $i$ such that
$X_i$ occurs in $E$. Consider such an $X_i$.
Since $X_i$ occurs in $E$, it must occur in at least one of $E_1$ and $E_2$. Let's assume that $X_i$ occurs in $E_1$; the case in which it occurs in $E_2$ is  entirely similar. If  $Y$ is the set of attributes of  $E_1(W,{\bf R}_1)$, then $X_i\subseteq Y$. Furthermore, the property of an expression  being monotone with respect to a witness function and a collection of relations is inherited by its subexpressions, so $E_1$ is monotone with respect to $W$ and ${\bf R}_1$. By induction hypothesis, $E_1(W,{\bf R}_1)$ is a global consistency witness of all relations $R_j$ occurring in it, hence
\begin{equation}
E_1(W,{\bf R}_1)[X_i]=R_i. \label{eqn:first}
\end{equation}
Also, since $W(E_1(W,{\bf R}_1), E_2(W,{\bf R}_2))$ is a consistency witness for $E_1(W,{\bf R}_1)$ and $E_2(W,{\bf R}_2)$, we have that 
\begin{equation}
W(E_1(W,{\bf R}_1), E_2(W,{\bf R}_2))[Y]=E_1(W,{\bf R}_1). \label{eqn:second}
\end{equation}
By putting everything together, we have that
\begin{align*}
& W(E_1(W,{\bf R}_1),E_2(W,{\bf R}_2))[X_i] \\
& = W(E_1(W,{\bf R}_1), E_2(W,{\bf R}_2))[Y][X_i] \\
& = E_1(W,{\bf R}_1)[X_i] \\
& = R_i,
\end{align*}
where in the first equality we used Proposition~\ref{lem:easyfacts1} and 
the fact that $X_i \subseteq Y$,
in the second we used~\eqref{eqn:second}, and the third is~\eqref{eqn:first}. 
This completes the proof of Proposition~\ref{prop:monotone-general}. 
}

\section{Acyclic Hypergraphs} \label{sec:acyclic}
As mentioned in Section \ref{sec:prelims}, if $R_1,\ldots,R_m$ are standard relations that are globally consistent, then they are pairwise consistent, but the converse does  not always hold. For example, consider the
\emph{triangle} schema $\{A,B\}, \{B,C\}, \{C,A\}$ and the  standard relations $R_1(A,B)=\{ (0,0), (1,1)\}$, $R_2(B,C)=\{ (0,1),(1,0)\}$,
$R_3(C,A)=\{(0,0), (1,1)\}$. It is easy to check that these  standard relations are pairwise consistent; however, they are not globally consistent since $((R_1\Join R_2)\Join R_3) = \emptyset$.  Beeri et al.\ \cite{BeeriFaginMaierYannakakis1983} characterized the schemas for which every collection of pairwise consistent standard relations is globally consistent by showing that these are precisely the \emph{acyclic} schemas. To give the precise definition of an acyclic schema, we need to first introduce some basic notions about hypergraphs.

Let $H=(V,F)$ be a hypergraph.
\begin{itemize}
    \item A \emph{path in $H$} is a sequence
    $X_1,\ldots,X_k$ of 
    hyperedges of $H$ such that $X_i\cap X_{i+1}\not = \emptyset$, for every $i$ with $1\leq i <k$.  
    In this case, we say that there is a \emph{path from $X_1$ to $X_k$}.
    \item  We say that a set $G$ of hyperedges of $H$ is \emph{connected} if
    for every two distinct hyperedges $X$ and $X'$ of $G$, there is a path from $X$ to $X'$.
 \item A \emph{connected component} of $H$ is a maximal connected set of hyperedges of $H$
    \item We say that  $H$ is \emph{connected} if the set $F$ of the hyperedges of $H$ is connected (in other words, $H$ has a single connected component); otherwise, we say that $H$ is \emph{disconnected}.
  \item $H$ is \emph{reduced} if no hyperedge of $H$ is properly contained in some other hyperedge of $H$. 
  \item The \emph{reduction} of $H$ is the hypergraph $(V,F')$, where $F'$ consists of the hyperedges in $F$ that are not properly contained  in some other hyperedge in  $F$.
  \item If $U\subseteq V$, then the \emph{restriction of $H$ on $U$}, denoted $H\restriction U$, is
  the reduction of the hypergraph
  $(U,\{X\cap U:~ X \in F \} \setminus \{\emptyset\})$.
  \item Let $H$ be a reduced hypergraph and let $X,X'$ be two distinct hyperedges.  We say that $Y=X\cap X'$ is an \emph{articulation set}
  of $H$ if  the number of connected components of $H\restriction (V\setminus Y)$ is bigger than the number of connected components of $H$.
  \item Let $H$ be a reduced hypergraph. We say that $H$ is \emph{acyclic}   if the following condition holds: for every set 
  $U\subseteq V$, if $H\restriction U$ is 
connected and has at least two hyperedges, then it has an articulation set; otherwise, we say that $H$ is \emph{cyclic}.
\item $H$ is \emph{acyclic} if its reduction is acyclic. Such schemas are also known
as 
\emph{$\alpha$-acyclic} schemas. 
\end{itemize}

Admittedly, the notion  of
an acyclic hypergraph appears to be difficult to grasp when encountered for the first time.
Intuitively, it generalizes to hypergraphs the property that
 a graph is acyclic if and only if every connected component of it  with at least two hyperedges has an articulation point. 
It is well known that there a polynomial-time algorithm for testing if a hypergraph is acyclic; this algorithm is due
to Graham \cite{Graham1980} and, independently,  to Yu and Ozsoyoglu \cite{yu1979algorithm},  and it  is known as the GYO algorithm (see also \cite{DBLP:journals/jacm/Fagin83}).

There are  several different structural conditions that are equivalent to  acyclicity. We discuss two of these notions next.

Let $H=(V,F)$ be a hypergraph.
\begin{itemize}
\item The \emph{Gaifman graph} $G(H)$ of $H$ is the undirected graph with nodes the attributes of $H$ and such that there is an edge between two attributes of $H$ if and only if  both these attributes belong to one of the hyperedges of $H$.
\item $H$ is \emph{conformal} if  every clique of $G(H)$ is contained in one of the hyperedges of $H$.
\item $H$ is \emph{chordal} if every cycle of $G(H)$ of length at least $4$ has a \emph{chord}, i.e., there is an edge of $G(H)$ that is not an edge of the cycle.
\item $H$ has the \emph{running intersection property} if there is an ordering $Y_1,\ldots,Y_m$ of the hyperedges of $H$ such that for every $i\leq m$, there is a $j<i$ such that $(Y_1\cup \cdots \cup  Y_{i-1})\cap Y_i \subseteq Y_j$.
\end{itemize}
The proof of the next result can be found in
\cite{BeeriFaginMaierYannakakis1983}.
\begin{proposition} \label{pro:acyclicity-structural} For every hypergraph $H$, the following statements are equivalent:
\begin{enumerate}
    \item $H$ is acyclic.
    \item $H$ is conformal and chordal.
    \item $H$ has the running intersection property.
\end{enumerate}
\end{proposition}

The triangle schema $\{A,B\}, \{B,C\}, \{C,A\}$ is cyclic because it is not conformal (but it is chordal); 
the  \emph{4-cycle} schema $\{A,B\}, \{B,C\}, \{C,
D\}, \{D,A\}$ is cyclic because it is not chordal (but it is conformal). For every $n\geq 2$, the  \emph{$n$-path} schema  $P_n$ with hyperedges
$\{A_1,A_2\}, \{A_2,A_3\}, \ldots, \{A_{n},A_{n+1}\}$ is acyclic because it is both conformal and chordal.
One can also reason about these schemas using the running intersection property. Finally, consider the schema $\{A,B,C\}, \{C,D,E\}, \{E,F,A\}, \{A,C,E\}$. It has the running intersection property via the ordering
$\{A,B,C\}, \{A,C,E\}, \{C,D,E\}, \{E,F,A\}$, hence it is acyclic.

Beeri et al.\ \cite{BeeriFaginMaierYannakakis1983} showed that acyclicity can be characterized in terms of useful semantic properties, where by ``semantic'' we mean a property of the hypergraph whose definition involves also standard relations. Specifically, consider the following two properties.

Let $H$ be a hypergraph with $X_1,\ldots,X_m$ as its hyperedges.
\begin{itemize}
\item $H$ has the \emph{\ltgc~for standard relations}  if for every  collection $R_1(X_1),\ldots,R_m(X_m)$   of pairwise consistent standard relations over $H$, we have that this  collection is globally consistent.
\item $H$ \emph{admits a monotone join expression} if there is a \cjoin~expression $E$ such that 
\begin{enumerate}
\item[i] 
$E$  is monotone with respect to the standard join operation $\Join$ and every collection of pairwise consistent relations $R_1(X_1),\ldots,R_m(X_m)$;
\item[ii]   Every hyperedge $X_i$ of $H$ occurs in $E$.
    \end{enumerate}
    \end{itemize}
    As mentioned in Section \ref{sec:prelims}, if a collection of relations is globally consistent, then it is pairwise consistent. Thus, if a schema has the \ltgc~for standard relations, then global consistency coincides with pairwise consistency for relations over that schema. Furthermore, if a schema admits a monotone join expression, then, in view of Proposition \ref{prop:monotone-general}, this join expression can be used to construct witnesses to global consistency of collections of pairwise consistent relations.
    
With the notions of \ltgc~and monotone join expression at hand, the main result in Beeri et al.\ \cite{BeeriFaginMaierYannakakis1983} can be stated as follows.
\begin{theorem}[\cite{BeeriFaginMaierYannakakis1983}] \label{thm:BFMY}
For every hypergraph $H$, the following statements are equivalent:
\begin{enumerate}
\item $H$ is acyclic.
\item $H$ has the \ltgc~for standard relations.
\item $H$ admits a monotone sequential join expression.
\end{enumerate}
\end{theorem}
 Thus, the triangle schema and the $4$-cycle schema do not have the \ltgc~for standard relations and do not admit a monotone join expression, but each $n$-path schema does, and so does the schema $\{A,B,C\}, \{C,D,E\}, \{E,F,A\}, \{A,C,E\}$.

In \cite{AK25}, the following question was investigated: do the result in Theorem \ref{thm:BFMY} extend from standard relations to $\mathbb K$-relations, where $\mathbb K$ is an arbitrary positive commutative monoid? 

The first realization in \cite{AK25} was that acyclicity is a necessary, but not always sufficient, condition for the \ltgc~to hold for $\mathbb K$-relations. More formally, we say that  a hypergraph  $H=\{X_1,\ldots,X_m\}$ has the \emph{\ltgc~for $\mathbb K$-relations} if every collection $R_1(X_1),\ldots,R_m(X_m)$ of pairwise consistent $\mathbb K$-relations is also globally consistent. 

We can now state the precise result about the necessity of acyclicity.

\begin{theorem}[\cite{AK25}] \label{thm:acyc-necessity}
The following statements are true:
\begin{enumerate}
\item For all positive commutative monoids $\mathbb K$ and hypergraphs $H$, if $H$
has the \ltgc~for $\mathbb K$-relations, then $H$ is acyclic.
\item There are positive commutative monoids $\mathbb K$ and hypergraphs $H$ such that $H$ is acyclic and does \emph{not} have the \ltgc~for $\mathbb K$-relations. In particular, this is the case for every numerical semigroup $\mathbb K$ other than the bag monoid ${\mathbb N}=(N,+, 0)$ (e.g., take ${\mathbb K}=(\langle 3, 5\rangle, +,0)$) and for the $3$-path  hypergraph  $P_3$. 
\end{enumerate}
\end{theorem}

Furthermore, in \cite{AK25}, a class of positive commutative monoids $\mathbb K$ was identified for which the acyclicity of a hypergraph $H$ is a sufficient condition for $H$ to have  the \ltgc~for
$\mathbb K$-relations.
 \begin{itemize}
\item If $m$ and $n$ are positive integers,  we say that a positive commutative monoid~$\mathbb K$ has the $m\times n$ \emph{\ftp}~if
 for
 every
 column~$m$-vector~$b = (b_1,\ldots,b_m) \in K^m$ with entries in~$K$ 
 and
 every row~$n$-vector~$c = (c_1,\ldots,c_n) \in K^n$ 
 with entries in~$K$ 
 such that~$b_1 + \cdots + b_m = c_1 + \cdots + c_n$ holds,  there
is  an~$m \times n$
matrix~$D = (d_{ij} : i \in [m], j \in [n]) \in K^{m \times n}$ with
entries in~$K$ whose 
rows sum to $b$ and whose columns sum to~$c$,
 i.e.,~$d_{i1} + \cdots + d_{im} = b_i$ for all~$i \in [m]$
 and~$d_{1j} + \cdots + d_{mj} = c_j$ for all~$j \in [n]$.
\item 
We say that a positive commutative monoid~$\mathbb K$ has the \emph{\ftp} if $\mathbb K$ has the $m\times n$ \ftp~for every pair $(m,n)$ of positive integers.
\end{itemize}
As shown in \cite{AK25}, examples of monoids with the transportation property include the Boolean monoid ${\mathbb B}=(\{0,1\}, \vee, 0)$, the bag monoid
${\mathbb N}=N,+,0)$, the monoids
${\mathbb T}= (R\cup\{\infty\}, \min, \infty)$ and  ${\mathbb V}=([0,1], \max, 0)$, and the power set monoid ${\mathbb P}(A)= (\mathcal{P}(A), \cup,\emptyset)$, for every set $A$.  In contrast,  no numerical semigroup other than the bag monoid has the transportation property.

To state the extension of Theorem \ref{thm:BFMY} to $\mathbb K$-relations, we need some additional notions, which were introduced in \cite{DBLP:journals/sigmod/AtseriasK25}.

Let $\mathbb K$~be a positive commutative monoid, let $X_1,\ldots,X_m$ be a schema,
and let $E$ be a \cjoin~expression over $X_1,\ldots,X_m$.

\begin{itemize}
\item 
First, recall from Section \ref{sec:cons}, that  $E$~is monotone with respect to a consistency witness function $W$ on $\mathbb K$  and a collection  $R_1(X_1),\ldots,R_m(X_m)$~of $\mathbb K$-relations if for every sub-expression $E_1 \cj E_2$ of $E$, we have that the $\mathbb K$-relations $E_1(W,R_1,\ldots,R_m)$
and $E_2(W,R_1,\ldots,R_m)$ are consistent.
\item  We say that $E$ is \emph{monotone on $\mathbb K$} if there is a consistency witness function $W$ on $\mathbb K$ such that $E$ is monotone with respect to  $W$  and every collection $R_1(X_1), \ldots,R_m(X_m)$ of pairwise consistent $\mathbb K$-relations.
    \item  
     We say that $E$ is \emph{strongly monotone on $\mathbb K$} if for every consistency witness function $W$ on $\mathbb K$, we have that $E$ is monotone  with respect to $W$
     and every collection $R_1(X_1), \ldots,R_m(X_m)$ of pairwise consistent $\mathbb K$-relations.
     \end{itemize}
     
Finally, we define what it means for a schema to admit a monotone and a strongly monotone \cjoin~expression.
     \begin{itemize}
    \item A schema  $X_1,\ldots,X_m$ \emph{admits a monotone \cjoin~expression on $\mathbb K$} if there is a  \cjoin-expression $E$ over $X_1,\ldots,X_m$ such that  $E$ is monotone on $\mathbb K$ and every hyperedge $X_i$ occurs in $E$.
 \item A schema $X_1,\ldots,X_m$ \emph{admits a strongly monotone \cjoin~expression on $\mathbb K$} if there is a  \cjoin-expression $E$ over $X_1,\ldots,X_m$ such that $E$ is strongly monotone on $\mathbb K$
        and  every hyperedge $X_i$ occurs in $E$.
\end{itemize}

We can now state one of the main  results from \cite{AK25,DBLP:journals/sigmod/AtseriasK25}.
\begin{theorem} [\cite{AK25,DBLP:journals/sigmod/AtseriasK25}]\label{thm:BFMY-general} Let $\mathbb K$ be a positive commutative monoid that has the transportation property.  For every hypergraph $H$,  
 the following statements are equivalent:
\begin{enumerate} \itemsep=0pt
\item $H$ is acyclic.
\item $H$ has the local-to-global consistency property for $\mathbb K$-relations.
\item $H$ admits a monotone sequential \cjoin-expression on $\mathbb K$.
\item $H$ admits a strongly monotone sequential \cjoin~expression on $\mathbb K$.
\end{enumerate}
\end{theorem}

Theorem \ref{thm:BFMY-general} yields Theorem \ref{thm:BFMY} of Beeri et al.\ \cite{BeeriFaginMaierYannakakis1983} by taking $\mathbb K$ to be the Boolean monoid $\mathbb B$. In fact, Theorem \ref{thm:BFMY-general} yields something stronger:  the standard join $\Join$ can be replaced in Theorem \ref{thm:BFMY} by an arbitrary consistency witness function for standard relations. In effect, this means that the \emph{only} property of the standard join needed in Theorem \ref{thm:BFMY} is that the standard join is a consistency witness function for standard relations.

As shown in \cite{AK25}, the transportation property actually characterizes the positive commutative monoids $\mathbb K$ for which every acyclic hypergraph has the \ltgc~for~$\mathbb K$-relations.
Furthermore, the transportation property turns out to be  equivalent to the \emph{inner consistency} property, which is defined as follows in \cite{AK25}.
\begin{itemize}
    \item 
    Two $\mathbb K$-relations $R(X)$ and $S(Y)$ are \emph{inner consistent} if
$R[X\cap Y]=S[X \cap Y]$. 

\item We say that $\mathbb K$ has the \emph{inner consistency property} if whenever two $\mathbb K$-relations are inner consistent, then they  are also consistent.
\end{itemize}
Note that, using Proposition \ref{lem:easyfacts1}, it is easy to verify that if $R$ and $S$ are consistent $\mathbb K$-relations, then they are also inner consistent. 
Consequently, for  monoids with the inner consistency property, the notions of consistency and inner consistency coincide.

\begin{theorem} [\cite{AK25, DBLP:journals/sigmod/AtseriasK25}]
\label{thm:TP}
 Let $\mathbb K$ be a  positive commutative monoid. Then the following statements are equivalent:
\begin{enumerate} 
\item $\mathbb K$ has the \ftp.

\item $\mathbb K$ has the inner consistency property.

\item Every acyclic hypergraph has the \ltgc~for $\mathbb K$-relations.

\item The $3$-path hypergraph $P_3$ has the \ltgc~for $\mathbb K$-relations.

\item Every acyclic hypergraph admits
a  monotone sequential  \cjoin-expression on $\mathbb K$
\item Every acyclic hypergraph admits
a strongly monotone sequential \cjoin-expression on $\mathbb K$.

\end{enumerate}

\end{theorem}

\section{Beta Acyclic Hypergraphs} \label{sec:beta}
Their good structural and semantic properties notwithstanding, acyclic hypergraphs suffer from the drawback that acyclicity is not a hereditary property, that is to say, a sub-hypergraph of an acyclic hypergraph need not be acyclic. For example, the hypergraph
$\{A,B,C\}, \{A,B\}, \{B,C\}, \{C,A\}$~is acyclic but it contains as a sub-hypergraph the triangle hypergraph $\{A,B\}$, $\{B,C\}$, $\{C,A\}$, which is cyclic. 

Motivated by the preceding considerations, Fagin \cite{DBLP:journals/jacm/Fagin83} introduced the following notion.
\begin{itemize}
    \item A hypergraph $H$ is \emph{$\beta$-acyclic} if every sub-hypergraph of $H$ is acyclic; otherwise, $H$ is \emph{$\beta$-cyclic}.
    \end{itemize}
    For example,  the hypergraph $H^*$ with hyperedges $\{A,B,C\}, \{A,B\},  \{A,C\}$ is $\beta$-acyclic.
There is a polynomial-time algorithm for testing if a hypergraph is $\beta$-acyclic
(see \cite{DBLP:journals/jacm/Fagin83}).
    
Fagin \cite{DBLP:journals/jacm/Fagin83}  established several different characterizations of $\beta$-acyclicity, including some that involve the absence of cycles of certain types. We spell out one of these, since we will refer to it later on.
\begin{itemize}
\item A \emph{weak $\beta$-cycle} in a hypergraph $H$ is a sequence $Y_1,A_1,Y_2,A_2,\ldots,Y_k,A_k,Y_{k+1}$ such that $k\geq 3$ and the following properties hold:
\begin{enumerate}
\item $Y_1,\ldots,Y_k$ are distinct hyperedges of $H$ and $Y_{k+1}=Y_1$;
\item $A_1,\ldots,A_k$ are distinct nodes of $H$;
\item For $1\leq i\leq k$, the node $A_i$ is in  $Y_i \cap Y_{i+1}$.
\item For $1\leq i\leq k$,  the node $A_i$ is not in any hyperedge $Y_j$ other than $Y_i$ 
and $Y_{i+1}$.
\end{enumerate}
\end{itemize}
\begin{theorem} [\cite{DBLP:journals/jacm/Fagin83}] \label{thm:weak-beta}
For every hypergraph $H$, the following statements are equivalent:
\begin{enumerate}
\item $H$ is $\beta$-acyclic.
\item $H$ has no weak $\beta$-cycles.
\end{enumerate}
\end{theorem}
For example, 
consider the schema $\{A,B,C\}, \{C,D,E\}, \{E,F,A\}, \{A,C,E\}$, which was shown earlier to be acyclic. The sequence 
$\{A,B,C\}, C, \{C,D,E\}, E, \{E,F,A\}, A,\{A,B,C\}$ is a weak $\beta$-cycle, hence this schema is $\beta$-cyclic.

Because of the hereditary nature of $\beta$-acyclicity, the characterizations of acyclicity in Section \ref{sec:acyclic} give rise to characterizations of $\beta$-acyclicity. Thus, the following result holds.

\begin{corollary}
Assume that $\mathbb K$ is a positive commutative monoid that has the transportation property.  For every hypergraph $H$,  
 the following statements are equivalent:
\begin{enumerate} \itemsep=0pt
\item $H$ is $\beta$-acyclic.
\item Every sub-hypergraph of $H$~has the local-to-global consistency property for $\mathbb K$-relations.
\item Every sub-hypergraph of $H$~admits a monotone sequential \cjoin-expression on $\mathbb K$.
\item Every sub-hypergraph of $H$~admits a strongly monotone sequential \cjoin~expression on $\mathbb K$.
\end{enumerate}
\end{corollary}

\section{Gamma-Acyclic Hypergraphs} \label{sec:gamma}

Fagin \cite{DBLP:journals/jacm/Fagin83} introduced and studied $\gamma$-acyclic hypergraphs, which form a proper subclass of the class of $\beta$-acyclic hypergraphs. As with $\beta$-acyclic hypergraphs, there are several equivalent formulations of the notion 
of a $\gamma$-acyclic hypergraph in terms of absence of cycles of certain types, including the following one.
\begin{itemize}
\item A \emph{weak $\gamma$-cycle} in a hypergraph $H$ is a sequence $Y_1,A_1,Y_2,A_2,\ldots,Y_k,A_k,Y_{k+1}$ such that $k\geq 3$ and the following properties hold:
\begin{enumerate}
\item $Y_1,\ldots,Y_k$ are distinct hyperedges of $H$ and $Y_{k+1}=Y_1$;
\item $A_1,\ldots,A_k$ are distinct nodes of $H$;
\item For $1\leq i\leq k$, the node $A_i$ is in  $Y_i \cap Y_{i+1}$; 
\item For $i=1,2$, the node $A_i$ is not in any hyperedge  $Y_j$ other than $Y_i$ 
and $Y_{i+1}$.
\end{enumerate}
\item A hypergraph $H$ is \emph{$\gamma$-acyclic} if $H$ has no weak $\gamma$-cycle; otherwise, $H$ is \emph{$\gamma$-cyclic}.
\end{itemize}
Clearly, every sub-hypergraph of a $\gamma$-acyclic hypergraph is $\gamma$-acyclic as well. Observe that
the only difference between a weak $\beta$-cycle and a weak $\gamma$-cycle is in the fourth condition of the definitions of these notions: the requirement that the node $A_i$ belongs only to the hyperedges $Y_i$ and $Y_{i+1}$ holds for every $i\leq k$
in the case of a weak $\beta$-cycle, while it  holds for $i=1,2$ in the case of a weak $\gamma$-cycle. In particular, every weak $\beta$-cycle is also a weak $\gamma$-cycle; consequently, every $\gamma$-acyclic hypergraph is also $\beta$-acyclic. The converse, however, is not true. To see this, consider the hypergraph $H^*$ with hyperedges $\{A,B,C\}, \{A,B\},  \{A,C\}$, which, as pointed out in Section \ref{sec:beta}, is $\beta$-acyclic. Clearly, the sequence $\{A,B\}, B, \{A, B, C\}, C, \{A,C\}, A, \{A,B,C\}$ is a weak $\gamma$-cycle, hence this hypergraph is $\gamma$-cyclic.  It is also easy to see that for every $n\geq 2$, the $n$-path hypergraph $P_n$ with hyperedges
$\{A_1,A_2\},\ldots,\{A_n,A_{n+1}\}$ is $\gamma$-acyclic.
Note that there is a polynomial-time algorithm, due to D' Atri and Moscarini \cite{DBLP:conf/pods/DAtriM84},  for testing whether or not a hypergraph is $\gamma$-acyclic (see also \cite[Section 9.4]{DBLP:journals/jacm/Fagin83}). 

There are several different structural characterizations of $\gamma$-acyclic hypergraphs, including one which is due to
Brault-Baron \cite{DBLP:journals/csur/Brault-Baron16} and which involves
the hypergraph $H^*$ above. To describe this characterization, we need to introduce the following basic notion.
\begin{itemize}
\item Let $H=(V,F)$ be a hypergraph and let $S$ be a subset of the set $V$ of the nodes of $H$. The
\emph{induced hypergraph} $H[S]$ is the hypergraph with hyperedges $\{X\cap S: X \in E\}\setminus \{\emptyset\}$.
\end{itemize}
\begin{proposition} [\cite{DBLP:journals/csur/Brault-Baron16}]\label{prop:brault-baron}
For every hypergraph $H$, the following statements are equivalent:
\begin{enumerate}
\item $H$ is $\gamma$-acyclic.
\item $H$ is $\beta$-acyclic and 
there do not exist  three nodes
$A,B,C$ of $H$ such that  the hypergraph $H^*$ with hyperedges $\{A,B,C\}, \{A,B\},  \{A,C\}$  is a sub-hypergraph of $H[\{A,B,C\}]$.
\end{enumerate}
\end{proposition}
In fact, Brault-Baron \cite{DBLP:journals/csur/Brault-Baron16} defines a $\gamma$-acyclic hypergraph to be a hypergraph that satisfies the second condition in Proposition \ref{prop:brault-baron}; he then  shows that this condition is equivalent to the D' Atri and Moscarini algorithm \cite{DBLP:conf/pods/DAtriM84} producing the empty hypergraph, hence this condition is equivalent to Fagin's \cite{DBLP:journals/jacm/Fagin83}
definition of $\gamma$-acyclicity.

Fagin \cite{DBLP:journals/jacm/Fagin83} established that $\gamma$-acyclic hypergraphs have certain desirable semantic properties; the main such property involves the notion of a \emph{connected} join expression.
\begin{itemize}
\item A join expression $E$  is \emph{connected} if for each of its sub-expressions $(E_1\Join E_2)$, there is an attribute that appears in both $E_1$ and $E_2$.
\end{itemize}
In particular, if $E$ is a sequential join expression $(((\cdots (Y_1\Join Y_2) \Join \cdots  )\Join Y_{k-1})\Join Y_k)$, then $E$ is connected if and only if  for every with $1\leq i \leq k$, we have that $(Y_1\cup \cdots \cup Y_{i-1})\cap Y_i \not = \emptyset$.
For example, the sequential join expression
$((\{A_1,A_2\}\Join \{A_2,A_3\})\Join \{A_3,A_4\})$ is connected, while the sequential join expression 
$(\{A_1,A_2\}\Join \{A_3,A_4\})$ is not connected.

We can now state the main semantic characterization of  $\gamma$-acyclicity, obtained in \cite{DBLP:journals/jacm/Fagin83}.
\begin{theorem}[\cite{DBLP:journals/jacm/Fagin83}] \label{thm:fagin-gamma}
For every hypergraph $H$, the following statements are equivalent:
\begin{enumerate}
\item $H$ is $\gamma$-acyclic.
\item Every connected sequential join expression over $H$ is monotone.
\end{enumerate}
\end{theorem}
The notion of a connected join expression extends to the notion of a connected \cjoin~expression in a straightforward way.
\begin{itemize}
\item A \cjoin~expression $E$  is \emph{connected} if for each of its sub-expressions $(E_1\cj E_2)$, there is an attribute that appears in both $E_1$ and $E_2$.
\end{itemize}
In particular, if $E$ is a sequential \cjoin~expression $(((\cdots (Y_1\cj Y_2) \cj \cdots )\cj Y_{k-1} )\cj Y_k)$, then $E$ is connected if for every with $1\leq i \leq k$, we have that $(Y_1\cup \cdots \cup Y_{i-1})\cap Y_i \not = \emptyset$.

It is now natural to ask: does Theorem \ref{thm:fagin-gamma} extend and how does it extend to $\gamma$-acyclic hypergraphs, \cjoin~expressions, and $\mathbb K$-relations, where $\mathbb K$ is  a positive commutative  monoid? 

Let $\mathbb K$ be a positive commutative monoid and let $E$
be a \cjoin~expression.
Recall that $E$ is monotone on  $\mathbb K$ if there is a consistency witness function $W$ on $\mathbb K$ such that $E$ is monotone with respect to  $W$  and every collection $R_1(X_1), \ldots,R_m(X_m)$ of pairwise consistent $\mathbb K$-relations.
Recall further that
 $E$ is strongly monotone on $\mathbb K$ if $E$ is monotone  with respect to every consistency witness function $W$ on $\mathbb K$ and every collection $R_1(X_1), \ldots,R_m(X_m)$ of pairwise consistent $\mathbb K$-relations.
    
    In what follows in this section, we will establish the following results:
    \begin{enumerate}
\item If $H$ is a hypergraph 
such that every connected sequential \cjoin-expression over $H$ is monotone on $\mathbb K$, then $H$ is $\gamma$-acyclic.
\item If $\mathbb K$ has the transportation property and $H$ is $\gamma$-acyclic, then 
every connected sequential \cjoin-expression over $H$ is strongly monotone on $\mathbb K$,
    \end{enumerate}


\begin{theorem}\label{thm:gamma-necessity}
For all positive commutative monoids $\mathbb K$ and hypergraphs $H$, if $H$ is
such that every connected sequential \cjoin-expression over $H$ is monotone on $\mathbb K$, then $H$ is $\gamma$-acyclic.
\end{theorem}
\begin{proof}
We will prove the contrapositive, that is,  if the hypergraph
$H$ is not $\gamma$-acyclic, then there is a connected sequential \cjoin~expression  $E$ over $H$ such that $E$ is not monotone on $\mathbb K$.  So, assume that $H$ is not $\gamma$-acyclic. We distinguish two cases, namely, the case in which $H$ is not  $\beta$-acyclic and the case in which $H$ is $\beta$-acyclic.

\noindent{\emph{Case 1:}} Assume that $H$ is not $\beta$-acyclic. By the definition of $\beta$-acyclicity, there is a sub-hypergraph $H'$ of $H$ that is cyclic. Moreover, we may assume that $H'$ is connected, since if every connected component of $H'$ were acyclic, then it is easy to see that $H'$ would be acyclic as well. Let $X_1,\ldots,X_k$ be a list of the hyperedges of $H'$. Since $H'$ is cyclic,  the first part of Theorem \ref{thm:acyc-necessity} implies that there are $\mathbb K$-relations $R_1(X_1),\ldots,R_k(X_k)$ that are pairwise consistent, but not globally consistent. Since $H'$ is connected, there is a sequence
$Y_1,\ldots,Y_t$ of not necessarily distinct sets of attributes such that the following hold: (a) each $Y_j$ is one of the hyperedges  of $H'$,  i.e., $Y_j=X_{i_j}$, where $i_j\in \{1, \ldots, k\}$; (b) each hyperedge $X_i$ of $H'$ appears in the sequence $Y_1,\ldots, Y_t$; and (c) for every $j$ with $1\leq j< t$, we have that $Y_j\cap Y_{j+1}\not = \emptyset$. Let $E$ be the sequential \cjoin~expression
$(((\cdots (Y_1\cj Y_2) \cj \cdots )\cj Y_{t-1} )\cj Y_t)$. Then $E$ is a connected \cjoin~expression because $Y_j\cap Y_{j+1}\not = \emptyset$ holds, for every $j$ with $1\leq j < t-1$. We now claim that there is no consistency witness function $W$ on $\mathbb K$ such that $E$ is monotone with respect to $W$. To see this, let $W$ be a consistency witness function on $\mathbb K$ and consider the  pairwise consistent $\mathbb K$-relations $R_1(X_1),\ldots,R_k(X_k)$. If $E$ were monotone with respect to $W$, then, by Proposition \ref{prop:monotone-general}, the $\mathbb K$-relation $E(W,R_{i_1},\ldots,R_{i_t})$    
is a global consistency witness for the relations $R_{i_1},\ldots,
R_{i_t}$, hence it is a global consistency witness for the relations $R_1,\ldots,R_k$ since every hyperedge $X_i$ appears in  the sequence $Y_1,\ldots, Y_t$. This, however, is a contradiction since the relations $R_1,\ldots,R_k$ are not globally consistent.

\noindent{\emph{Case 2:}} Assume that $H$ is $\beta$-acyclic.
Since $H$ is not $\gamma$-acyclic, Proposition
\ref{prop:brault-baron} implies that 
there are three attributes $A$, $B$, $C$ and three hyperedges $Y_1, Y_2, Y_3$ of $H$  such that $Y_1\cap \{A,B,C\} = \{A,B\}$, $Y_2 \cap \{A,B,C\} = \{A,C\}$, and $Y_3 \cap \{A,B,C\} = \{A.B,C\}$.   Let $D_1, D_2, D_3$ be the remaining sets of attributes in $Y_1,Y_2, Y_3$, respectively; thus,  $Y_1=\{A,B\}\cup D_1$, $Y_2= \{A,C\}\cup D_2$, $Y_3 = \{A,B,C\}\cup D_3$. 
Let $E$ be the sequential \cjoin~expression $((Y_1\cj Y_2)\cj Y_3))$, which is clearly connected since $\{A,B,C\} \subseteq (Y_1\cup\ Y_2)\cap Y_3$.

We now claim that there is no consistency witness function $W$ on $\mathbb K$ such that $E$ is monotone with respect to $W$. 
Consider the following two $\mathbb K$-relations
$R_1(A,B,D_1), R_2(A,C,D_2)$:
\begin{itemize}
\item 
$R_1(f,f,{\bf f}_1) = a$, $R_1(f,t,{\bf f}_1)=a$, and $R_1(x,y, {\bf z}_1)=0$, for all other values,
where ${\bf f}_1=(f,\ldots,f)$  and the length of the tuple ${\bf f}_1$ is the cardinality of the set $D_1$,
\item 
$R_2(f,f,{\bf f}_2) = a$, $R_2(f,t,{\bf f}_2)=a$, and $R_2(x,y,{\bf z}_2)=0$, for all other values,
where ${\bf f}_2=(f,\ldots,f)$,
and the length of the tuple ${\bf f}_2$ is the cardinality of the set $D_2$.
\end{itemize}
Consider the following two $\mathbb K$-relations $S_1(A,B,C,D_1,D_2)$
and $S_2(A,B,C,D_1,D_2)$:
\begin{itemize}
    \item $S_1(f,f,f,{\bf f})=a$,
    $S_1(f,t,t,{\bf f})=a$, and 
    $S_1(x,y,z,{\bf w})=0$, for all other values.
    \item $S_2(f,f,t,{\bf f})=a$, 
    $S_2(f,t,f,{\bf f})=a$, and $S_2(x,y,z,{\bf w}) =0$, for all other values, where
    ${\bf f}$ is a tuple of $f$'s of length equal to the cardinality of the set $D_1\cup D_2$.
\end{itemize}
 It is easy to verify that 
 $S_1(A,B,C,D_1,D_2)$
and $S_2(A,B,C,D_1,D_2)$
 are two different consistency witnesses for the relations $R_1(A,B,D_1)$ and $R_2(A,C)$.
 
Now, let $W$ be an arbitrary witness function on $\mathbb K$.
Since the relations $R_1$
and $R_2$ are consistent, we have that the $\mathbb K$-relation
$W(R_1,R_2)$ is a consistency witness function for $R_1$ and $R_2$.
We distinguish the following two sub-cases.

\noindent{\emph{Sub-case 1:}}
$W(R_1,R_2)[ABC]=S_1[ABC]$. In this case,
let $R_3(A,B,C,D_3)$ be the $\mathbb K$-relation
such that 
 $R_3(f,f,t,{\bf f}_3)=a$, 
    $R_3(f,t,f,{\bf f}_3)=a$, and $R_3(x,y,z,{\bf w}) =0$, for all other values, where
    ${\bf f}_3$ is a tuple of $f$'s of length equal to the cardinality of the set $D_3$.
Observe that $R_3[ABC]=S_2[ABC]$.
    Since $S_2$ is a consistency witness for $R_1$ and $R_2$, we have that  the relations $R_1,R_2,R_3$ are pairwise consistent. However, the $\mathbb K$-relations $W(R_1,R_2)=S_1$
and $R_3$ are not consistent, 
since $S_1[ABC] \not =  S_2[ABC]$ and $S_2[ABC]= R_3[ABC]$.

\noindent{\emph{Sub-case 2:}}
$W(R_1,R_2)[ABC]\not =S_1[ABC]$. In this case,
let $R_3(A,B,C,D_3)$ be the $\mathbb K$-relation
such that 
 $R_3(f,f,f,{\bf f}_3)=a$, 
    $R_3(f,t,t,{\bf f}_3)=a$, and $R_3(x,y,z,{\bf w}) =0$, for all other values, where
    ${\bf f}_3$ is a tuple of $f$'s of length equal to the cardinality of the set $D_3$.
    Observe that $R_3[ABC]=S_1[ABC].$
Since $S_1$ is a consistency witness for $R_1$ and $R_2$, we have that  the relations $R_1,R_2,R_3$ are pairwise consistent. 
However, the $\mathbb K$-relations $W(R_1,R_2)$
and $R_3$ are not consistent, 
since $W(R_1,R_2)[ABC]\not = S_1[ABC]=R_3[ABC]$.
\end{proof}

Since the hypergraph $H^*$ with hyperedges
$\{A,B,C\}, \{A,B\},  \{A,C\}$  is $\gamma$-acyclic, Theorem \ref{thm:gamma-necessity} implies that there is a connected sequential \cjoin~expression over $H^*$ that is not monotone on $\mathbb K$. Actually, the proof of Theorem \ref{thm:gamma-necessity} implies that $((\{A,B\} \cj \{A,C\}) \cj \{A,B,C\})$ is such an expression. We give a direct proof of this fact in the Appendix, since that direct proof led us to the proof of Theorem \ref{thm:gamma-necessity}.

The preceding Theorem \ref{thm:gamma-necessity} asserts that $\gamma$-acyclicity is a necessary condition for a hypergraph $H$ to have the property that every connected sequential \cjoin~expression over $H$ is monotone on $\mathbb K$, where $\mathbb K$ is an arbitrary positive commutative monoid.
The second main result in this section asserts that if
$\mathbb K$ has the transportation property, then
$\gamma$-acyclicity is a sufficient condition for a hypergraph to have the property that every connected sequential \cjoin~expression over $\mathbb K$ is strongly monotone on $\mathbb K$.

\begin{theorem} \label{thm:gamma-sufficient}
Let $\mathbb K$ be a positive commutative monoid that has the transportation property. If $H$ is a $\gamma$-acyclic hypergraph, then every connected sequential \cjoin~expression over $H$ is strongly monotone on $\mathbb K$.
\end{theorem}
\begin{proof}
 We will establish the contrapositive, that is, we will show that if $H$ does not have the property that 
  every connected sequential \cjoin~expression over $H$ is strongly monotone on $\mathbb K$, 
 then $H$ is $\gamma$-cyclic.

 Assume that $H$ lacks the above property. Then there exist a connected sequential \cjoin~expression $E=
  (((\cdots (X_1\cj X_2) \cj \cdots )\cj X_{m-1} )\cj X_m)$ over $H$, 
 a consistency witness function $W$ on $\mathbb K$, and a collection $R_1(X_1),\ldots,R_m(X_m)$ of pairwise consistent $\mathbb K$-relations such that  $E$  is not monotone w.r.t.\ to $W$ and  $R_1(X_1),\ldots,R_m(X_m)$. In turn, this means that there is some index $j<m$ such that the $\mathbb K$-relation $E_j(W,R_1,\ldots,R_j)$ is not consistent with the $\mathbb K$-relation $R_{j+1}$, where $E_j$ is the sequential \cjoin~expression 
 $(((\cdots (X_1\cj X_2) \cj \cdots )\cj X_{j-1} )\cj X_j)$.
 Let $j$ be the smallest index with this property; thus, if $i< j$, then the  $\mathbb K$-relation $E_i(W,R_1,\ldots,R_i)$ is consistent with the $\mathbb K$-relation $R_{i+1}$.

Let $Y=(X_1\cup \cdots \cup X_j)\cap X_{j+1}$. 
Note that $Y\not = \emptyset$, since 
$E$ is a connected sequential \cjoin~expression.
Note also that $X_1\cup \cdots \cup X_j$ is the set of attributes of 
$E_j(W,R_1,\ldots,R_j)$, while $X_{j+1}$ is the set of attributes of  $R_{j+1}$. 
Since $\mathbb K$ has the transportation property, Theorem \ref{thm:TP} implies that $\mathbb K$~has the inner consistency property, which means that if two $\mathbb K$-relations $Q_1(X_1)$ and $Q_2(X_2)$ are inner consistent (i.e., $Q_1[X_1\cap X_2]=Q_2[X_1\cap X_2])$, then they are consistent. Thus, since the  $\mathbb K$-relations
$E_j(W,R_1,\ldots,R_j)$ and $R_{j+1}$
are not consistent,
 it follows that
$E_j(W,R_1,\ldots,R_j)$ and $R
_{j+1}$ are not inner consistent, hence
$E_j(W,R_1,\ldots,R_j)[Y]\not = R_{j+1}[Y].$

\smallskip

\noindent{\emph{Claim 1:}} For every  $k\leq j$, we have that $Y\not \subseteq X_k$.

To establish the claim, assume that 
$Y \subseteq X_k$, for some $k\leq j$. By the pairwise consistency of $R_1,\ldots,R_m$, we have that  $R_k$ and $R_{j+1}$ are consistent $\mathbb K$-relations, hence
$$R_k[X_k\cap X_{j+1}]=R_{j+1}[X_k\cap X_{j+1}].$$
Since $Y\subseteq X_k\cap X_{j+1}$, we have that
$R_k[Y]=R_{j+1}[Y]$. Proposition \ref{prop:monotone-general} and the minimality assumption about $j$  imply that 
the $\mathbb K$-relation $E_j(W,R_1,\ldots,R_j)$ is a global consistency witness for $R_1,\ldots,R_j$. Since $k\leq j$, we have that $E_j(W,R_1,\ldots,R_j)[X_k]=R_k$. Furthermore, since $Y\subseteq X_k$, we have
that $E_j(W,R_1,\ldots,R_j)[Y]=R_k[Y]$, hence
$E_j(W,R_1,\ldots,R_j)[Y]=R_{j+1}[Y]$; this  contradicts the earlier finding that $E_j(W,R_1,\ldots,R_j)[Y]\not = R_{j+1}[Y]$,
hence  Claim 1 has been established.

So, we now know that for every $k\leq j$, we have that $Y\not \subseteq X_k$. Choose an index $k\leq j$ so that the cardinality $|X_k\cap Y|$ of the set 
 $X_k\cap Y$ is the largest of the cardinalities $|X_i\cap Y|$ of the sets $X_i\cap Y$, $1\leq i\leq j$. 
 
 Pick a node $A_1$ such that $A_1\in Y\setminus X_k$. 
 Since $Y\subseteq X_1\cup \cdots \cup X_j$, there is some $i\leq j$ such that $A_1\in X_i$. Furthermore, since 
 $E=
  (((\cdots (X_1\cj X_2) \cj \cdots )\cj X_{m-1} )\cj X_m)$ is a connected sequential \cjoin~expression over $H$, it is easy to see that  $\{X_1,\ldots,X_j\}$ 
 is a connected set of hyperedges of $H$. Let $p$ be the length of the shortest path within  $\{X_1,\ldots,X_j\}$ from $X_k$ to a hyperedge containing $A_1$. This means that there is a sequence
 $S_1,\ldots,S_p$ of sets with the following properties:
 \begin{enumerate}
     \item Each set $S_i$ is one of the hyperedges $X_1,\ldots,X_j$.
     \item $S_1=X_k$.
     \item $A_1\in S_p$.
     \item $S_i\cap S_{i+1}\not = \emptyset$, $1\leq i\leq p$.
     \item $p$ is the smallest possible number for which  a sequence $S_1,\ldots,S_p$ with properties (1)-(4) exists.
 \end{enumerate}
 The minimality of $p$ implies that $S_1,\ldots,S_p$ are distinct hyperedges of $H$.
 By the maximality of $k$, we have that $|S_p\cap Y|\leq |X_k\cap Y|$.
 Since $A_1\in S_p\cap Y$ and $A_1\not \in X_k\cap Y$, there must exist a node $A_2\in (X_k\cap Y)\setminus (S_p\cap Y)$.
 Let $n$ be the largest number such that $A_2\in S_n$ and $1\leq n< p$. This number $n$ exists because $S_1=X_k$.
By property (4) of the sequence $S_1\ldots,S_p$, there are nodes $B_i$, $n\leq i< j$  such that $B_i\in S_i\cap S_{i+1}$.
Consider now the sequence

\centerline{
$(S_p,A_1,X_{j+1},A_2,S_n,B_n,S_{n+1},B_{n+1},\ldots, B_{p-1},S_p)$.}

\smallskip

 \noindent{\emph{Claim 2:}}
$(S_p,A_1,X_{j+1},A_2,S_n,B_n,S_{n+1},B_{n+1},\ldots, B_{p-1},S_p)$ is a weak 
 $\gamma$-cycle in $H$.

We have to verify that this sequence satisfies the conditions defining a weak $\gamma$-cycle, which were spelled out in the beginning of this section. 
\begin{enumerate}
 \item The hyperedges $S_p, X_{j+1}, S_n, S_{n+1},\ldots, S_{p-1}$ are distinct, since
 the hyperedges $S_1,\ldots,S_p$ are distinct and also are among the hyperedges $X_1,\ldots,X_j$.
 \item The nodes $A_1,A_2,B_n,B_{n+1},\ldots ,B_{p-1}$ are  distinct for the following reasons: first,
 $A_1\not = A_2$ because $A_1\not \in X_k\cap Y$ and $A_2\in X_k\cap Y$; second, $A_1$ is different from $B_n,B_{n+1}, \ldots, B_{p-1}$, since $A_1 \not \in S_i$ for $n\leq i< p$; $A_2$ is different from $B_n,B_{n+1}, \ldots, B_{p-1}$ by the choice of $n$ as the largest number such that $A_2\in S_n$ and the fact that $B_i\in S_i\cap S_{i+1}$; third, the $B_i$'s are distinct, else the path would have been shorter.
 \item First, $A_1\in S_p\cap X_{j+1}$ by the choice of $S_p$ and the fact that
 $A_1\in Y\subseteq X_{j+1}$; second, $A_2\in X_{j+1}\cap S_n$ by the choice of $A_2$; third, $B_i\in S_i\cap S_{i+1}$ by the choice of the $B_i$'s.
\item $A_1$ is not in any of the hyperedges $S_n,S_{n+1}, \ldots, S_{p-1}$ by the properties of the path $S_1,\ldots, S_p$; finally, $A_2$ is not in any of the hyperedges $S_{n+1}m\ldots,S_p$ by the choice of $A_2$.
 \end{enumerate}
 This completes the proof of Claim 2, hence the hypergraph $H$ is $\gamma$-cyclic.
\end{proof}

The reader familiar with Fagin's paper \cite{DBLP:journals/jacm/Fagin83} will undoubtedly notice that the proof of Theorem \ref{thm:gamma-sufficient} has a very similar structure to the proof of the result that
if a hypergraph $H$ is $\gamma$-acyclic, then every connected sequential join expression is monotone (see  \cite[pages 539-540]{DBLP:journals/jacm/Fagin83}). The main difference is that here, instead of the standard join,  we use an arbitrary consistency witness function and, thus, obtain a stronger result about
annotated relations and arbitrary consistency witness functions. Furthermore, it is of the essence that the annotations come from a monoid $\mathbb K$ that has the transportation property (which certainly the Boolean monoid has) because, in the first part of the proof before Claim 1, we used in a crucial way that $\mathbb K$ has the inner consistency property (which is equivalent to the transportation property).

As an immediate consequence of Theorems \ref{thm:gamma-necessity} and \ref{thm:gamma-sufficient}, we obtain the following result.

\begin{corollary} \label{cor:gamma-characterize}
Let $\mathbb K$ be a positive commutative monoid that has the transportation property. For
every hypergraph $H$, the following statements are equivalent:
\begin{enumerate}
\item $H$ is $\gamma$-acyclic.
\item  Every connected sequential \cjoin~expression over $H$ is strongly monotone on $\mathbb K$.
\item Every connected sequential \cjoin-expression over $H$ is monotone on $\mathbb K$.
\end{enumerate}
\end{corollary}
\begin{proof}
The implication $(1)\Longrightarrow (2)$ follows from Theorem \ref{thm:gamma-sufficient}; the implication
$(2)\Longrightarrow (3)$ follows from the definitions;
the implication $(3) \Longrightarrow (1)$ follows from
Theorem \ref{thm:gamma-necessity}.
\end{proof}

Finally, we characterize  the transportation property in terms of $\gamma$-acyclicity.

\begin{theorem} \label{thm:transp-gamma}
 Let $\mathbb K$ be a  positive commutative monoid. Then the following statements are equivalent:
\begin{enumerate}
\item ${\mathbb K}$~has the transportation property.
\item Every acyclic hypergraph $H$ has the property that every connected sequential \cjoin~expression over $H$ is strongly monotone on $\mathbb K$.
\item Every acyclic hypergraph has the property that every connected sequential \cjoin-expression over $H$ is monotone on $\mathbb K$.
\item The $3$-path  hypergraph $P_3$ has the property that every connected sequential \cjoin-expression over $P_3$ is monotone on $\mathbb K$.
\end{enumerate}
 \end{theorem}
 \begin{proof}
 The implications $(1)\Longrightarrow (2)$ follows 
 from Theorem
 \ref{thm:gamma-sufficient}.
 The implication $(2)\Longrightarrow (3)$ follows from the definitions. 
 The implication $(3)\Longrightarrow (4)$ follows from the fact that the $3$-path hypergraph $P_3$ is acyclic. Finally, towards  the implication $(4)\Longrightarrow (1)$,
 let $P_3$ be the $3$-path hypergraph 
with hyperedges $\{A_1,A_2\}, \{A_2,A_3\}, \{A_3,A_4\}$. Assume that $P_3$ has the property that every connected sequential c-join~expression over $P_3$ is monotone on $\mathbb K$.
Let $E$ be the sequential c-join~expression $((\{A_1,A_2\} \cj \{A_2,A_3\})\cj  \{A_3,A_4\})$ over $H$ and let $W$ be a consistency witness function on $\mathbb K$ such that $P_3$ is monotone with respect to $W$ and every three
pairwise consistent $\mathbb K$-relations $R_1(A_1, A_2), R_2(A_2,A_3), R_2(A_3,A_4)$. But then $P_3$ has the \ltgc~property for $\mathbb K$-relations, because
if $R_1(A_1, A_2), R_2(A_2,A_3), R_2(A_3,A_4)$ are three pairwise consistent
$\mathbb K$-relations, then, by Proposition \ref{prop:monotone-general}, the $\mathbb K$-relation
$E(W,R_1,R_2,R_3)$ is a global consistency witness for the $\mathbb K$-relations $R_1(A_1, A_2), R_2(A_2,A_3), R_2(A_3,A_4)$.
Consequently, by Theorem \ref{thm:TP}, the monoid $\mathbb K$ has the transportation property.
\end{proof}

\section{Concluding Remarks}
In this paper, we showed that the main desirable semantic properties of $\gamma$-acyclic hypergraphs and standard relations extend to desirable semantic properties of $\gamma$-acyclic hypergraphs and annotated relations, as long as the annotations come from a positive commutative monoid possessing the transportation property. It can also be shown that other desirable  semantic properties of $\gamma$-acyclicity, such as join dependencies and lossless joins  (see Fagin \cite{DBLP:journals/jacm/Fagin83}), have suitable extensions to $\gamma$-acyclic hypergraphs and annotated relations; the notions of  a consistency witness function and a \cjoin~expression are used in these results.

Going back to acyclic hypergraphs, Beeri et al.\ \cite{BeeriFaginMaierYannakakis1983} characterized acyclic hypergraphs in terms of semijoin programs and full reducers, which make it possible to evaluate queries efficiently in a distributed setting. It remains an open problem to determine whether or not  there is an abstract notion of a semijoin program that can be used to extend these characterizations to acyclic schemas and annotated relations.

\paragraph{Acknowledgments} Atserias was partially supported by grant                     
no. PID2022-138506NB-C22 (PROOFS BEYOND) and the Severo Ochoa and María         
de Maeztu Program for Centers and Units of Excellence in R\&D (CEX\-2020-001084\
-M) of                                                                          
the AEI, and the CERCA and ICREA                                                
Academia Programmes of the Generalitat.

\bibliography{biblio}

\newpage

\section*{Appendix}

Let $\mathbb K$ be a positive commutative monoid.
We give a self-contained proof that  hypergraph the $H^*$ does not have the property that every connected sequential \cjoin~expression is monotone on $\mathbb K$.

\begin{proposition} \label{prop:H*}
If $\mathbb K$ is a positive commutative monoid and  $H^*$ is the hypergraph
with hyperedges $\{A,B,C\}, \{A,B\}, \{A,C\}$, then  the connected sequential
\cjoin~expression
$$((\{A,B\} \cj \{A,C\}) \cj \{A,B,C\}).$$ 
is not monotone on $\mathbb K$.
\end{proposition}
\begin{proof}
To show that the sequential
\cjoin~expression
$((\{A,B\} \cj \{A,C\})\cj \{A,B,C\})$ is not monotone on $\mathbb K$, we have to show that if $W$ is an arbitrary consistency witness function on $\mathbb K$, then there are three $\mathbb K$-relations
$R_1(A,B)$, $R_2(A,C)$, $R_3(A,B,C)$  that are pairwise consistent,
but the $\mathbb K$-relations
$W(R_1,R_2)$ and $R_3$ are not consistent.

Let $t$ and $f$ be two values such that $t\not = f$, and
let $a$ be an element in the universe of the monoid $\mathbb K$ such that $a\not =0$. Consider the
$\mathbb K$-relations
$R_1(A,B)$ and $R_2(A,C)$ such that
\begin{itemize}
\item 
$R_1(f,f) = a$, $R_1(f,t)=a$, and $R_1(x,y)=0$, for all other pairs.
\item 
$R_2(f,f) = a$, $R_2(f,t)=a$, and $R_2(x,y)=0$, for all other pairs.
\end{itemize}
It is easy to verify that the following two $\mathbb K$-relations $S_1(A,B,C)$
and $S_2(A,B,C)$ are consistency witnesses for the relations $R_1(A,B)$ and $R_2(A,C)$:
\begin{itemize}
    \item $S_1(f,f,f)=a$,
    $S_1(f,t,t)=a$, and 
    $S_1(x,y,z)=0$, for all other triples.
    \item $S_2(f,f,t)=a$, 
    $S_2(f,t,f)=a$, and $S_2(x,y,z) =0$, for all other triples.
\end{itemize}
Let $W$ be an arbitrary witness function on $\mathbb K$.
Since the relations $R_1$
and $R_2$ are consistent, we have that the $\mathbb K$-relation
$W(R_1,R_2)$ is a consistency witness function for $R_1$ and $R_2$.
We now distinguish the following two cases.

\noindent{\emph{Case 1:}}
$W(R_1,R_2)=S_1$. In this case,
let $R_3(A,B,C)=S_2(A,B,C)$. 
Since $S_2$ is a consistency witness for $R_1$ and $R_2$, we have that  the relations $R_1,R_2,R_3$ are pairwise consistent. However, $W(R_1,R_2)=S_1$
and $R_3=S_2$ are not consistent, 
since $S_1$ and $S_2$ are different $\mathbb K$-relations
on the same set of attributes.

\noindent{\emph{Case 2:}}
$W(R_1,R_2)\not =S_1$. In this case,
let $R_3(A,B,C)=S_1(A,B,C)$. 
Since $S_1$ is a consistency witness for $R_1$ and $R_2$, we have that  the relations $R_1,R_2,R_3$ are pairwise consistent. However, $W(R_1,R_2)$
and $R_3=S_1$ are not consistent, 
since $W(R_1,R_2)$ and $S_1$ are different $\mathbb K$-relations
on the same set of attributes.

Therefore, the connected sequential \cjoin~expression 
$((\{A,B\} \cj \{A,C\} )\cj \{A,B,C\})$ is not monotone on $\mathbb K$.  
\end{proof}

\end{document}